\documentclass[article,notitlepage,amsthm,amsmath,amssymb]{revtex4-2}  
\usepackage{color}
\usepackage[dvipsnames]{xcolor}
\usepackage{graphicx}
\usepackage{subcaption}
\usepackage{amsmath}
\usepackage{amssymb}
\usepackage{amsthm}
\usepackage[colorlinks = true,linkcolor = blue,urlcolor = blue,citecolor = blue,anchorcolor = blue]{hyperref}
\usepackage{pgfplots}
\usepackage{comment}

\usepackage{xy}
\xyoption{matrix}
\xyoption{frame}
\xyoption{arrow}
\xyoption{arc}

\usepackage{ifpdf}
\ifpdf
\else
\PackageWarningNoLine{Qcircuit}{Qcircuit is loading in Postscript mode.  The Xy-pic options ps and dvips will be loaded.  If you wish to use other Postscript drivers for Xy-pic, you must modify the code in Qcircuit.tex}
\xyoption{ps}
\xyoption{dvips}
\fi

\entrymodifiers={!C\entrybox}

\newcommand{\bra}[1]{{\left\langle{#1}\right\vert}}
\newcommand{\ket}[1]{{\left\vert{#1}\right\rangle}}
\newcommand{\qw}[1][-1]{\ar @{-} [0,#1]}
\newcommand{\qwx}[1][-1]{\ar @{-} [#1,0]}

\newcommand{\gate}[1]{*+<.6em>{#1} \POS ="i","i"+UR;"i"+UL **\dir{-};"i"+DL **\dir{-};"i"+DR **\dir{-};"i"+UR **\dir{-},"i" \qw}

\newcommand{\control}{*!<0em,.025em>-=-<.2em>{\bullet}}

\newcommand{\ctrl}[1]{\control \qwx[#1] \qw}

\newcommand{\targ}{*+<.02em,.02em>{\xy ="i","i"-<.39em,0em>;"i"+<.39em,0em> **\dir{-}, "i"-<0em,.39em>;"i"+<0em,.39em> **\dir{-},"i"*\xycircle<.4em>{} \endxy} \qw}
\newcommand{\Qcircuit}{\xymatrix @*=<0em>}

\newcommand{\Gate}[1]{\textsc{#1}}

\newcommand{\cnotgate}{\Gate{CNOT}} 
\newcommand{\czgate}{\Gate{CZ}}
\newcommand{\cczgate}{\Gate{CCZ}}
 
\newcommand{\swapgate}{\Gate{SWAP}} 
\newcommand{\hgate}{\Gate{H}}
\newcommand{\xgate}{\Gate{X}}

\newcommand{\zgate}{\Gate{Z}}

\newcommand{\FF}{\mathbb{F}}

\newtheorem{lemma}{Lemma}

\newcommand{\eq}[1]{Eq.~(\ref{eq:#1})}
\renewcommand{\sec}[1]{\hyperref[sec:#1]{Section~\ref*{sec:#1}}}
\newcommand{\ssec}[1]{\hyperref[ssec:#1]{Subsection~\ref*{ssec:#1}}}
\newcommand{\fig}[1]{\hyperref[fig:#1]{Figure~\ref*{fig:#1}}}
\newcommand{\tab}[1]{\hyperref[tab:#1]{Table~\ref*{tab:#1}}}
\newcommand{\lem}[1]{\hyperref[lem:#1]{Lemma~\ref*{lem:#1}}}
\newcommand{\propos}[1]{\hyperref[propos:#1]{Proposition~\ref*{propos:#1}}}
\newcommand{\thm}[1]{\hyperref[thm:#1]{Theorem~\ref*{thm:#1}}}
\newcommand{\alg}[1]{\hyperref[alg:#1]{Algorithm~\ref*{alg:#1}}}

\newcommand{\fortyeightqubittime}{0.00257947} 
\newcommand{\ninetysixqubittime}{4.16629} 

\newcommand{\ra}{\rangle}
\newcommand{\la}{\langle}

\newcommand{\ba}{\begin{array}}
\newcommand{\ea}{\end{array}}

\newcommand{\be}{\begin{equation}}
\newcommand{\ee}{\end{equation}}

\pgfplotsset{compat=1.18}

\begin{document}

\title{Fast classical simulation of Harvard/QuEra IQP circuits}

\author{Dmitri Maslov}
\affiliation{IBM Quantum, IBM T. J. Watson Research Center, Yorktown Heights, NY 10598, USA}
\author{Sergey Bravyi}
\affiliation{IBM Quantum, IBM T. J. Watson Research Center, Yorktown Heights, NY 10598, USA}
\author{Felix Tripier}
\affiliation{IonQ, 4505 Campus Drive, College Park, MD 20740, USA}
\author{Andrii Maksymov}
\affiliation{IonQ, 4505 Campus Drive, College Park, MD 20740, USA}
\author{Joe Latone}
\affiliation{IonQ, 4505 Campus Drive, College Park, MD 20740, USA}

\begin{abstract}
Establishing an advantage for (white-box) computations by a quantum computer against its classical counterpart is currently a key goal for the quantum computation community.  A quantum advantage is achieved once a certain computational capability of a quantum computer is so complex that it can no longer be reproduced by classical means, and as such, the quantum advantage can be seen as a continued negotiation between classical simulations and quantum computational experiments.

A recent publication (Bluvstein et al., Nature 626:58-65, 2024) introduces a type of Instantaneous Quantum Polynomial-Time (IQP) computation complemented by a $48$-qubit (logical) experimental demonstration using quantum hardware.  The authors state that the ``simulation of such logical circuits is challenging'' and project the simulation time to grow rapidly with the number of $\cnotgate$ layers added, see Figure 5d/bottom therein.  However, we report a classical simulation algorithm that takes only $0.00257947$ seconds to compute an amplitude for the $48$-qubit computation, which is roughly $10^3$ times faster than that reported by the original authors.  Our algorithm is furthermore not subject to a significant decline in performance due to the additional $\cnotgate$ layers.  We simulated these types of IQP computations for up to $96$ qubits, taking an average of $4.16629$ seconds to compute a single amplitude, and estimated that a $192$-qubit simulation should be tractable for computations relying on Tensor Processing Units. 
\end{abstract}

\maketitle

\section{Introduction}\label{sec:Intro}
Quantum computational technology is fast approaching the regime where quantum computations become difficult to simulate by classical means.  A few noteworthy examples of white-box quantum computations designed primarily to demonstrate an advantage over classical computations include Boson Sampling \cite{aaronson2011computational}, quantum supremacy \cite{preskill2012quantum, aaronson2016complexity, bouland2019complexity, aharonov2023polynomial}, and Instantaneous Quantum Polynomial-Time (IQP) circuits \cite{shepherd2009temporally}.  Benefits of the latter include conceptual simplicity, moderate requirements on hardware (e.g., available gate library), and suitability for demonstration using early fault-tolerant quantum computers since all required logical gates can be implemented transversally for a suitable error correction code \cite{paletta2023robust}.  Indeed, making supremacy circuits fault-tolerant is likely to require magic state distillation which is not practical for the time being.  This, however, also highlights a potential downside of the IQP demonstrations---the relative simplicity of the IQP circuits, including the layered structure of the IQP transformation being well-defined compared to essentially random supremacy computations enabling classical simulation of 50-qubit IQP circuits in a few minutes on a laptop computer \cite{codsi2022classically}.  As a result, one may expect to require a higher number of qubits to participate in an IQP computation before it is no longer possible to simulate classically compared to supremacy.  Indeed, \cite{dalzell2020many} estimates the number of qubits necessary to achieve an advantage by IQP computations at $208$, in contrast to what will likely be about $60$ (otherwise, a number $q$ large enough that a set of $2^q$ different amplitudes may not all be simultaneously stored in memory) for a sufficiently deep supremacy computation, for which a state vector style simulation \cite{pednault2019leveraging} will likely (no free lunch) be among the best possible.

A certain type of IQP circuit has been considered recently \cite{bluvstein2023logical} as a basis for demonstrating a classically difficult quantum computation.  The authors claim that ``demonstrated logical circuits are approaching the edge of exact simulation methods'' \cite{bluvstein2023logical}.  However, we believe the quantum computation considered, including its extensions to a larger circuit depth or a higher number of qubits, can be simulated quickly using classical hardware.  

The focus of our work is on the development of a classical simulation of IQP circuits considered in \cite{bluvstein2023logical} and the demonstration of the ability to simulate such circuits for a reasonably large number of qubits $n$, substantively exceeding that required for a supremacy-style simulation with comparable classical simulation effort.  Specifically, we develop an algorithm that consists of $O\left(2^{n/3}\right)$ simulations of a $\frac{2n}{3}$-qubit  Clifford unitary, each taking time $O(n^3)$, to perform a strong simulation of the Harvard/QuEra circuit $HQ$ \cite{bluvstein2023logical}.  Our simulation computes the entire amplitude of $\bra{y}HQ\ket{0}$ in $\fortyeightqubittime$ seconds, in contrast to a weaker notion of the simulation that would replicate sampling from the probability distribution offered by the quantum computation $HQ$.

\subsection{IQP circuits}

An $n$-qubit IQP circuit can be defined as the three-stage computation $C\,{=}\,HDH$, where the first and last layers $H$ are $n$-fold products of Hadamard gates, and the middle layer $D$ performs a diagonal transformation $\ket{x_1,x_2,...,x_n} \mapsto (-1)^{f(x_1,x_2,...,x_n)} \ket{x_1,q_x,...,x_n}$ for some Boolean function $f$ over $n$ inputs $x_1,x_2,...,x_n$.  The task in the IQP computation $C$ is to observe a $\ket{y}$ given input $\ket{0}$, $\bra{y}C\ket{0}$.  Indeed, it suffices to set the input state to $\ket{0}$ since performing the IQP computation over an input state $\ket{x}$ results in $\bra{y}C\ket{x} = \bra{x \oplus y}C\ket{0}$ and is thus reducible to a different measurement outcome $x \oplus y$ (the modulo-2 sum is taken component-wise). 

While $f$ in the above definition of an IQP computation can be an arbitrary function, it suffices to consider $f$ computed by the EXOR polynomials of degree 3, which offers a problem whose complexity was shown to be \#P-hard \cite{ehrenfeucht1990computational, dalzell2020many}.  Ref. \cite{dalzell2020many} further showed the evidence that the simulation complexity by a non-deterministic classical algorithm must scale at least as $2^{n/2}$.  The algorithm reported in this paper beats the above lower bound; this, however, comes at no contradiction since the $HQ$ circuit is well-structured and does not correspond to the hardest instance of an IQP computation with the diagonal term $D$ computed by a degree-3 polynomial.

\subsection{Structure of Harvard/QuEra circuit}

The IQP circuit $HQ_k$ \cite{bluvstein2023logical} is built over $n=3m=3\,{\cdot}\,2^k$ inputs ($k{=}4$ in the largest experiment demonstrated).  Between the two layers of Hadamard gates, it implements a combination of linear Boolean transformation and phase computation.  The qubits are grouped into blocks of three, as dictated by the [[8,3,2]] error-correcting code employed to reduce computational errors.  The $2^k$ blocks are thought of as vertices of a $k$-dimensional Boolean cube to which the $\cnotgate$ gates apply transversely (i.e., in the sets of three).  The $2^k$ blocks are broken into two non-overlapping sets such that the graph distance over the Boolean cube in each set is $2$ (cube nodes are graph vertices, and cube edges are graph edges).  These two sets act as controls and targets for the transversal $\cnotgate$ gates.  At each $\cnotgate$ gate layer, the $\cnotgate$ gates apply along a yet-unexplored geometric dimension of the Boolean cube, until all dimensions are exhausted by step $k$.  This offers $k{+}1$ opportunities to apply diagonal $Z$-axis gates $\zgate$, $\czgate$, and $\cczgate$ inside each block of qubits between the $\cnotgate$ gate layers.  See \fig{hqcirc} for a schematic illustration of the $HQ_4$ circuit considered in \cite{bluvstein2023logical} as well as Figure 5C therein. We highlight that this circuit introduces an asymmetry in the application of the third diagonal $Z$-axis layer.

\begin{figure}[t]
    \centering
    \includegraphics[width=0.8\textwidth]{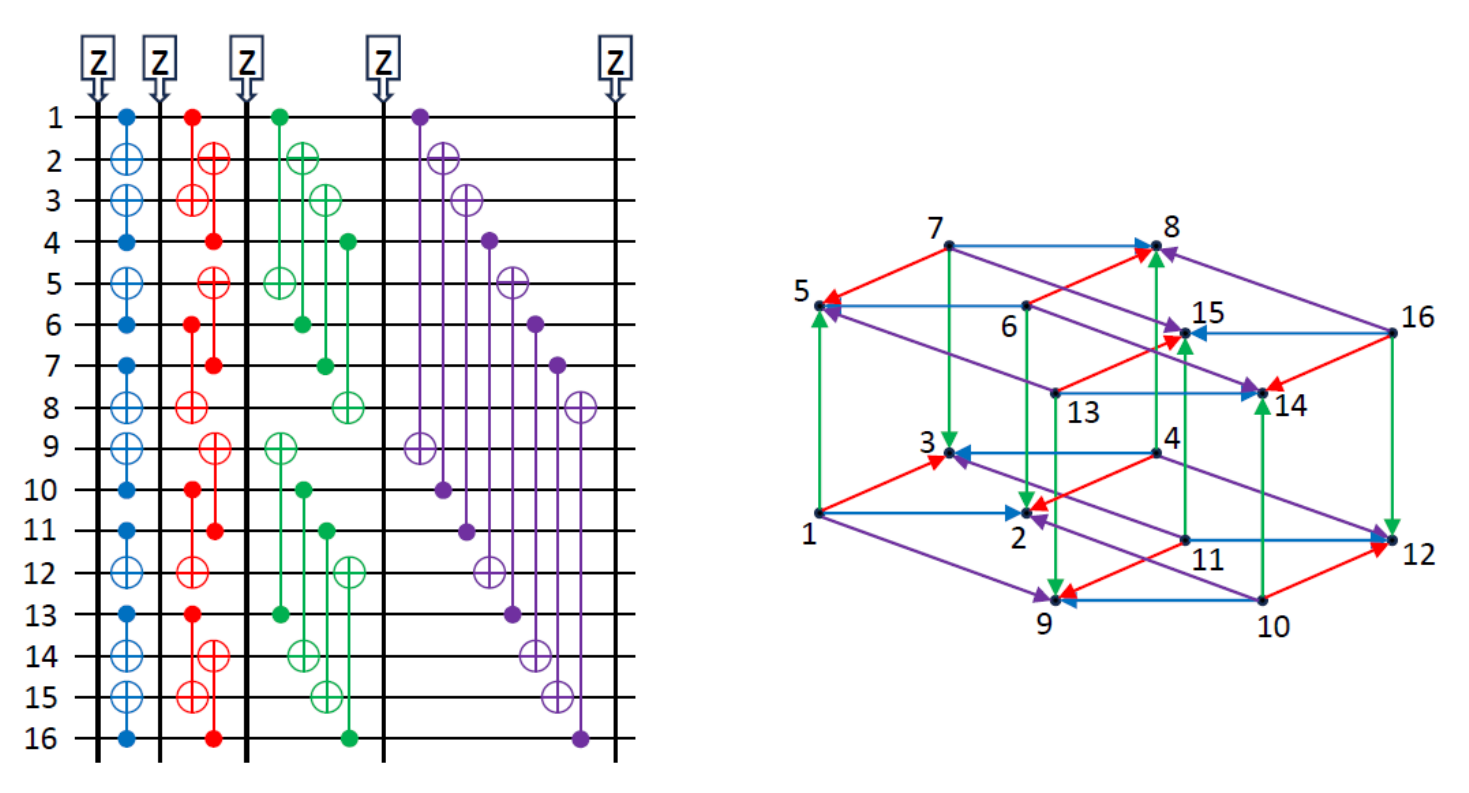}
    \caption{Illustration of the $HQ_4$ circuit.  Each block of three qubits is numbered $1..16$. The $\cnotgate$ gate colors match the colors of edges in the Boolean 4-dimensional cube illustrated on the right. The arrows point in the direction of targets of the respective $\cnotgate$ gates. The diagonal operators are applied in stages marked ``Z''. Not shown are two layers of Hadamard gates---one at the beginning and the other at the end of the computation.}
    \label{fig:hqcirc}
\end{figure}

\subsection{Types of simulation}
In this work, we distinguish two types of simulation of quantum computations.  First, noiseless `weak' or noiseless `direct' is a classical simulation that allows sampling bitstrings computed by a given quantum circuit.  Such a simulation is directly comparable to a {\em noiseless} quantum computer.  Since no current prototype quantum computer is noiseless, the above-defined `weak'/`direct' simulation is at least as powerful as the quantum computer itself.

Here, we focus on the notion of `strong' simulation, defined as the computation of the entire amplitude for a given desired observable $y$, $\bra{y}HQ_k\ket{0}$. 
This is enough to 
accomplish weak simulation, that is, produce a
sample from the output distribution of the circuit by employing known reductions from weak to strong simulation~\cite{bravyi2022simulate}.
Specifically, let $HQ_k(t)$ be the subcircuit of $HQ_k$ spanned by the first $t$ gates for some fixed ordering of gates. The gate-by-gate simulator of Ref.~\cite{bravyi2022simulate} 
takes as input a sample $y$ from the distribution $|\la y|HQ_k(t{-}1)|0^n\ra|^2$
and outputs a sample $y'$ from the distribution $|\la y'|HQ_k(t)|0^n\ra|^2$
by computing at most one amplitude of $HQ_k(t)$ and performing a simple classical processing. The amplitude computation is only needed when the $t$-th gate is Hadamard. 
Instead of starting the simulation from the first gate of $HQ_k$, we can fast forward to the final layer of Hadamards. Indeed, the output distribution of circuit $HQ_k$ terminated immediately before the final layer of Hadamards is uniform and can be sampled directly. Then the gate-by-gate simulator of~\cite{bravyi2022simulate}  would have to be called only to simulate $n$ Hadamard gates in the final layer of Hadamard. This requires $n$ amplitude computations~\footnote{We expect that the runtime is dominated by computing $O(1)$ amplitudes corresponding to subcircuits in which the final Hadamard layer is nearly full.
Indeed, removing a Hadamard gate is equivalent to setting the corresponding qubit to $0$ or $1$ throughout the circuit. Such a qubit does not support any superpositions and can be easily removed from the simulation by properly modifying the rest of the circuit.}.

\section{Our simulation}\label{sec:sim}

First, we express the $HQ_k$ circuit as a computation $HD_k^\prime H$.  The first layer of Hadamard gates performs the transformation
\begin{equation*}
\ket{0^n} \mapsto \frac{1}{\sqrt{2}^n}\sum_{x\in\{0,1\}^n} \ket{x}.  
\end{equation*}
Next comes the application of $D_k^\prime$, a diagonal operation together with a certain linear reversible function.  We write $D_k^\prime=$-$D_k$-CNOT- as a two-stage computation, where $D_k$ is the diagonal part and -CNOT- is a stage with $\cnotgate$ gates.  The -CNOT- part can be moved through the second layer of Hadamards by swapping targets and controls using the well-known circuit identity
\begin{equation*}
\Qcircuit @C=0.2em @R=0.2em @!{
& \ctrl{1}  & \gate{\hgate} & \qw \\ 
& \targ     & \gate{\hgate} & \qw }
\hspace{1em}\raisebox{-0.9em}{=}\hspace{1em}
\Qcircuit @C=0.2em @R=0.2em @!{
& \gate{\hgate}   & \targ       & \qw \\ 
& \gate{\hgate}   & \ctrl{-1}   & \qw }
\hspace{0.3em}\raisebox{-0.9em}{.}
\end{equation*}
Thus the -CNOT- stage (with inverted control/targets) can be applied directly to $\ket{y}$. Since this stage has $3k\cdot 2^{k-1}$ gates, the entire cost of simulating the linear reversible part is $O(n \log(n))$, and it is negligible. 

The leftover part $D_k$ contains $\zgate: \ket{x} \mapsto (-1)^x\ket{x}$, $\czgate: \ket{x,y} \mapsto (-1)^{xy}\ket{x,y}$, and $\cczgate: \ket{x,y,z} \mapsto (-1)^{xyz}\ket{x,y,z}$ gates, where $x$, $y$, and $z$ are either qubits transformed by the first layer of Hadamards or their linear sums, and thus compute a polynomial of degree at most $3$ into the phase.  We note that the Pauli-$Z$ gates can be flushed to the right side of the circuit using the following circuit diagram rules
\begin{equation*}
\Qcircuit @C=0.2em @R=0.2em @!{
& \gate{\zgate} & \ctrl{0}  & \qw  }
\hspace{0.3em}\raisebox{-0.0em}{$\mapsto$}\hspace{0.3em}
\Qcircuit @C=0.2em @R=0.2em @!{
& \ctrl{0}  & \gate{\zgate} & \qw }
\hspace{0.3em}\raisebox{-0.0em}{,}\hspace{0.7em}
\Qcircuit @C=0.2em @R=0.2em @!{
& \qw           & \ctrl{1}  & \qw \\ 
& \gate{\zgate} & \targ     & \qw }
\hspace{0.3em}\raisebox{-0.9em}{$\mapsto$}\hspace{0.3em}
\Qcircuit @C=0.2em @R=0.2em @!{
& \ctrl{1}  & \gate{\zgate} & \qw \\ 
& \targ     & \gate{\zgate} & \qw }
\hspace{0.3em}\raisebox{-0.0em}{, and}\hspace{0.7em}
\Qcircuit @C=0.2em @R=0.2em @!{
& \gate{\zgate} & \gate{\hgate}  & \qw }
\hspace{0.3em}\raisebox{-0.0em}{$\mapsto$}\hspace{0.3em}
\Qcircuit @C=0.2em @R=0.2em @!{
& \gate{\hgate}  & \gate{\xgate} & \qw }
\hspace{0.3em}\raisebox{-0.0em}{.}
\end{equation*}
This means that all $\zgate$ gates can be turned into bit flips that are easy to account for classically. This operation takes time $O(n \log(n))$ and thus its complexity is negligible.  This allows us to write the state vector evolution under the transformation $D_k$ as 
\begin{equation}\label{eq:2}
\frac{1}{\sqrt{2}^n}\sum_{x\in\{0,1\}^n} \ket{x} \mapsto \ket{\psi} = \frac{1}{\sqrt{2}^n}\sum_{x\in\{0,1\}^n} (-1)^{f(x)}\ket{x},
\end{equation}
where $f(x)$ is a degree-3 polynomial. Our goal is to compute an amplitude $\bra{y}HD_k^\prime H\ket{0} = \bra{y}H\ket{\psi}$ for a given $n$-tuple $y$.

To simulate the IQP circuit $HQ_k$ we first explore and develop the idea of covering sets mentioned in \cite{bravyi2019simulation}.  Specifically, let us say that a subset of qubits $S$ with $s$ qubits is a covering set if each degree-3 term in the polynomial $f(x)$ contains at least one variable from $S$.  Note that fixing all qubits in $S$ transforms a degree-3 polynomial to a degree-2 polynomial.  The degree-2 polynomial describes a $\czgate$ transformation (together with a layer of $\zgate$ gates, that can be efficiently accounted for as described in the previous paragraph).  Thus, each of the $2^s$ subspaces will be a Clifford circuit. Formally, the desired amplitude $\bra{y}HQ_k\ket{0}$ is obtained as the sum 
\begin{equation*} 
\bra{y}HQ_k\ket{0} = a_{0}(\ket{{+}{+}...{+}{+}{+}}) + a_{1}(\ket{{+}{+}...{+}{+}{-}}) + a_{2}(\ket{{+}{+}...{+}{-}{+}}) + ... + a_{2^{s}{-}1}(\ket{{-}{-}...{-}{-}{-}}),
\end{equation*}
where the sum goes across fixed variable assignments from the set $S$, and amplitudes $a_i$ are offered by the respective Clifford circuits of the form -H-CZ-Z-H- (-H-CZ-H- after Z gates are moved through Hadamards and turned into bit flips).  The circuit with the fixed set of qubits is indeed Clifford since the first layer of Hadamards transforms $\ket{+}/\ket{-}$ into $\ket{0}/\ket{1}$ correspondingly and each $\cczgate$ gate, by the construction of the set $S$, takes one of three inputs as a Boolean $0$ or $1$. The replacements $\cczgate(0,x,y) = Id$ and $\cczgate(1,x,y) = \czgate(x,y)$ remove all non-Clifford gates from the $HQ_k$ circuit.  Thus, amplitude $a_i$ can be found by simulating a Clifford circuit spanning $n{-}s$ qubits.  Our goal is to find the covering set $S$ with as few qubits as possible.

\begin{lemma}\label{lem:1}
For the $HQ_k$ circuit, the minimal covering set contains exactly $n/3$ qubits.
\end{lemma}
\begin{proof}
First, we establish the upper bound.  Name the qubits $1..3{\cdot}2^k$ top to bottom. The set $S:=\{3i{+}1|i=0..2^k{-}1\}$ is a covering set because, due to the $\cnotgate$ gate transversality, each $\cczgate$ takes as the first/top input EXOR of a subset of such variables.  Therefore, when written as a proper polynomial, each degree-3 term will contain a variable from the set $S$.  This set has $n/3$ qubits.

The lower bound on the number of variables in the covering set is $n/3$. To show it, we first observe that each 3-term product created at level $j{=}0..k$ (between the layers of $\cnotgate$ gates illustrated in \fig{hqcirc}, numbered $0$ to $k$) is generated again exactly once at all subsequent layers. This is because the next layer of $\cnotgate$ gates adds new variables to the linear sums that get multiplied, which adds new terms to the polynomial without removing any of the existing ones. It means that any product of primary variables that is first created at level $j$ will be present in the final reduced polynomial so long as $k{-}j$ is even, and not present in it when $k{-}j$ is odd.

If $k$ is even, the final phase polynomial will contain Boolean product terms $(3i{+}1){\cdot}(3i{+}2){\cdot}(3i{+}3)$ for all $i=0..2^k{-}1$ first generated at level $0$. This means that the set $S$ must contain at least one variable from each such term, of which there are $2^k$.  In other words, $s \geq n/3$. 

If $k$ is odd, each 6-tuple of qubits will generate polynomial $1{\cdot}2{\cdot}6 \oplus 1{\cdot}5{\cdot}3 \oplus 1{\cdot}5{\cdot}6 \oplus 4{\cdot}2{\cdot}3 \oplus 4{\cdot}2{\cdot}6 \oplus 4{\cdot}5{\cdot}3$ (shown is the polynomial expression for top $6$ qubits, it is similar of all other $6$-tuples), requiring a supporting set of size $2$, which can be established by observation. Since there are $n/6$ such non-overlapping sets, $s \geq 2\cdot n/6 = n/3$.
\end{proof}

\lem{1} allows to simulate $HQ_k$ by relying on $2^{n/3}$ simulations of a Clifford $2m$-qubit circuit of the form -H-CZ-H- (recall that $n\,{=}\,3m$).

\begin{lemma}\label{lem:cliff}
A $2m$-qubit Clifford circuit of the form -$\hgate^{2m}$-$\czgate$-$\hgate^{2m}$- resulting from slicing the $HQ_k$ circuit can be strongly simulated in time $O(m^3)$. 
\end{lemma}
The proof of this lemma as well as a detailed description of the simulation algorithm can be found in Ref.~\cite{bravyi2019simulation}, see Lemma~4 thereof. 
Our software \cite{Bravyi_Harvard_QuEra_Phase_Polynomial}, as tested, adopts the original implementation of the algorithm developed in Ref.~\cite{bravyi2019simulation}.  Below we describe a simplified algorithm exploiting the special structure of $HQ_k$ circuits.

As noted in \eq{2}, the state generated by the leading Hadamard and the diagonal layers
of the $HQ_k$ circuit has a form
\[
|\psi\ra =  \frac1{2^{n/2}} \sum_{x\in \{0,1\}^n} (-1)^{f(x)} |x\ra,
\]
where  $f\,{:} \, \FF_2^n \to \FF_2$ is a degree-$3$ Boolean polynomial.
Each triple of logical qubits encoded by the code $[[8,3,2]]$
can be naturally labeled by red, green, and blue colors~\cite{bluvstein2023logical}.
We will write $x\,{=}\,(x^R,x^G,x^B)$, where $x^c\,{\in}\, \{0,1\}^m$ is the restriction of $x$
onto the register spanned by all qubits with the color $c\,{\in}\, \{R,G,B\}$.
The polynomial $f(x)$ contains cubic monomials $x^R_i x^G_j x^B_\ell$ originating from $\cczgate$ gates and quadratic monomials $x^R_i x^G_j$, $x^R_i x^B_\ell$, $x^G_j x^B_\ell$ originating from $\czgate$ gates. 
For simplicity, we ignore linear terms in $f$. Recall that adding a degree-1 monomial such as $x^R_i$ to $f$ is equivalent to flipping the bit $R_i$ in the measured bit string. Thus degree-1 monomials in $f$ can be removed by a simple classical processing of the measured bit string. 
Monomials featuring more than one variable of the same color such as $x^R_1 x^G_2 x^G_3$ are prohibited since their generation requires $\cczgate$ or $\czgate$ gates that cannot be implemented transversally for the considered color code.
Assume for simplicity that red, green, and blue qubits span consecutive $m$-qubit blocks.
Consider any output bit string $y\,{\in}\, \{0,1\}^n$. 
Let $y^c\,{\in}\, \{0,1\}^m$ be the restriction of $y$ onto the register $c\,{\in}\, \{R,G,B\}$.
Applying the final Hadamard layer to $\psi$ gives
\[
\la y| HQ_k |0^n\ra =\la y|  H^{\otimes n} |\psi\ra =\frac1{2^{m}} \sum_{x^R \in \{0,1\}^m} \; (-1)^{y^R\cdot x^R} \la y^G y^B |H^{\otimes 2m} |GB(x^R)\ra,
\]
where 
\[
|GB(x^R)\ra = \frac1{2^{m}} \sum_{x^G,x^B\in \{0,1\}^m} \;
(-1)^{f(x^R,x^G,x^B)} |x^G x^B\ra
\]
is a $2m$-qubit state parameterized by $x^R$. Moreover, this state can be created using only Clifford gates
(Hadamards, $\czgate$, and $\zgate$) since $f(x^R,x^G,x^B)$ has degree-2 in the variables $x^G$, $x^B$ for a fixed $x^R$. Thus
one can compute  an amplitude  $\la y^G y^B |H^{\otimes 2m} |GB(x^R)\ra$ efficiently using
off-the-shelf Clifford simulators. Computing the amplitude $\la y| HQ_k |0^n\ra$
 then amounts to carrying out $2^m$
Clifford simulations on $2m$ qubits.
In fact, the Clifford circuit generating $|GB(x^R)\ra$ has a special structure that 
simplifies the simulation task. Namely,  fix a string $x^R\,{\in}\, \{0,1\}^m$.
We can write
\be
\label{f_fixed_xR}
f(x^R,x^G,x^B) = x^G \cdot \Gamma x^B + \delta^G \cdot x^G + \delta^B \cdot x^B 
\ee
for some matrix $\Gamma\,{\in}\, \{0,1\}^{m\times m}$ and some vectors $\delta^G,\delta^B \,{\in}\, \{0,1\}^m$ that depend on $x^R$ (we do not explicitly show this dependence to ease the notations).
Here the dots ($\cdot$) denote the inner product of binary vectors and $\Gamma x^B$ denotes matrix-vector multiplication.  Then
\[
\la y^G y^B |H^{\otimes 2m} |GB(x^R)\ra=
\frac1{2^{2m}} \sum_{x^G,x^B\in \{0,1\}^m}\;
(-1)^{x^G \cdot \Gamma x^B + \delta^G_y \cdot x^G  + \delta^B_y \cdot x^B},
\]
where $\delta^G_y = \delta^G+y^G$ and $\delta^B_y = \delta^B+y^B$.
The above sum can be computed analytically resulting in
\be
\label{Clifford_amplitude}
\la y^G y^B |H^{\otimes 2m} |GB(x^R)\ra=
\left\{
\ba{rcl}
\frac{(-1)^{\delta^B_y \cdot \Gamma^{-1}\delta^G_y}}{2^{\mathrm{rank}(\Gamma)}}
&& \mbox{if } 
\delta^B_y \,{\in}\, \mathrm{Row}(\Gamma)
\mbox{ and }
\delta^G_y \,{\in}\, \mathrm{Col}(\Gamma) 
 \\
0 && \mbox{else} \\
\ea
\right..
\ee
Here we write $\mathrm{Col}(\Gamma)$ and $\mathrm{Row}(\Gamma)$ for the linear subspace of $\{0,1\}^m$ spanned by the columns and rows of $\Gamma$ respectively.
Finally, by a slight abuse of notations, we write  $\Gamma^{-1}\delta^G_y$ for any solution $x^B\,{\in}\, \{0,1\}^m$ of the linear system
\be
\label{linear_system}
\Gamma x^B = \delta^G_y
\ee
(we do not assume that $\Gamma$ is invertible).  This system is feasible whenever $\delta^G_y\,{\in}\, \mathrm{Col}(\Gamma)$, which is the only case when Eq.~(\ref{Clifford_amplitude})
requires solving the system.
The inner product $\delta^B_y \cdot \Gamma^{-1} \delta^G_y=\delta^B_y\cdot x^B$ in Eq.~(\ref{Clifford_amplitude}) is the same for all solutions $x^B$ of the linear system Eq.~(\ref{linear_system}) due to the condition $\delta^B_y \,{\in}\, \mathrm{Row}(\Gamma)$.
Indeed, different solutions $x^B$ of the linear system Eq.~(\ref{linear_system}) differ by a vector from the nullspace of $\Gamma$.
Such vector is orthogonal to any row of $\Gamma$. Thus Eq.~(\ref{Clifford_amplitude}) is well defined, even if $\Gamma$ is not an invertible matrix.

To conclude, 
 computing amplitudes Eq.~(\ref{Clifford_amplitude}) requires only the standard linear algebra
over the binary field: computing the rank of a matrix and solving linear systems.
The time complexity scales as $O(m^3)$. The overall computation of
the amplitude $\la y| HQ_k |0^n\ra$  takes time $O(m^32^m)$.

\section{Details of the implementation}

Our C/C++ implementation is available publicly \cite{Bravyi_Harvard_QuEra_Phase_Polynomial}.  It is designed to simulate $HQ_k$ circuits with the number of qubits $n$ up to 96 ($k{=}5$).  Due to the demonstrated practical efficiency of the simulation for $n\,{\leq}\,96$, we believe the demonstration of the simulation for the next larger $n$, $n{=}192$, will not be necessary since $HQ_k$ circuits appear to be insufficiently difficult for establishing a quantum advantage and better alternatives exist.

When $n\,{\leq}\,96$, \lem{cliff} addresses circuits with up to $64$ qubits, which, given our implementation \cite{Bravyi_Harvard_QuEra_Phase_Polynomial}, conveniently coincides with the number of bits in a single machine word in modern computers.  Thus, a layer of single-qubit $\zgate$ operations can be described by one machine word.  We store the $\czgate$ gates in a $2m {\times} 2m$ Boolean array, which is slightly (by a factor of 4; this is because the inputs of the commuting self-inverse $\czgate$ gates come one from each of the sets $x^B$ and $x^G$ of cardinality $m$) wasteful.  However, we ensure that the computation stays in the processor caches, and thus we see the reduction of memory as unnecessary.  

We cycle through the bit patterns $x^R$ in Gray code order.  At each step, the matrix $\Gamma$ and vectors $\delta^G$, $\delta^B$ parameterizing the polynomial $f(x^R,x^G,x^B)$ according to  Eq.~(\ref{f_fixed_xR}) are updated. This update is facilitated by the fact that $\Gamma$ is a linear function of $x^R$. Thus updating the matrix $\Gamma$ upon flipping a single bit of $x^R$ requires a single addition of $2m\,{\times}\, 2m$ matrices. As noted above, this amounts to $2m$ bitwise operations since we represent each row of $\Gamma$ by a single $64$-bit integer~\footnote{The tested C++ implementation treats blue and green qubits on the same footing. Accordingly, the $\czgate$ circuit generating the state $|GB(x^R)\ra$ is represented by a binary matrix of size $2m\,{\times}\, 2m$. In contrast, the Python version of our simulator solves a linear system with the compact $m\,{\times}\, m$ matrix}.
Before executing a Clifford simulation we perform a straightforward verification ensuring that the contribution to the desired amplitude is non-zero. Namely, using the symmetries of the circuit $HQ_k$ one can show that the matrix $\Gamma$ is symmetric and the vector $x^R$ always belongs to the nullspace of $\Gamma$.  Accordingly, the row space and the column space of $\Gamma$ are the same.
The Clifford amplitude computed in Eq.(\ref{Clifford_amplitude}) is zero whenever $\delta^G_y \,{\cdot}\, x^R\,{=}\,1$ or $\delta^B_y \,{\cdot}\, x^R\,{=}\,1$.
Indeed, since $x^R$ belongs to the nullspace of $\Gamma$, it is orthogonal to any vector in the row space or the column space of $\Gamma$.
This simple check discards roughly $3/4$ of the Boolean patterns $x^R$ for which 
$\delta^G_y \,{\cdot}\, x^R\,{=}\,1$ or $\delta^B_y \,{\cdot}\, x^R\,{=}\,1$.

\section{Runtime and resource estimate}\label{sec:runtime}

\begin{figure}[t]
    \begin{tikzpicture}
        \begin{axis}[
            xlabel=Additional $\cnotgate$ layers,
            ylabel=Runtime per bitstring (sec),
            ymode=log,
            ymin=0.00001,
            every axis/.append style={font=\small},
            ]
        \addplot[color=teal,mark=*,mark size=3] coordinates {
            (0,2)
            (1,9)
            (2,200)
            (3,400000)
        };
        \addplot[color=blue,mark=*,mark size=3,mark options={solid}] coordinates {
            (0,0.00257947)
            (1,0.00289219)
            (2,0.00332019)
            (3,0.00358142)
            (4,0.00401678)
            (5,0.00428741)
        };
        \end{axis}
    \end{tikzpicture}
    \caption{Comparison of simulation runtimes with the number of additional layers added between \cite[Figure 5d bottom]{bluvstein2023logical} (teal, redrawn), and our simulation (blue).}
    \label{fig:fig5dcomparison}
\end{figure}
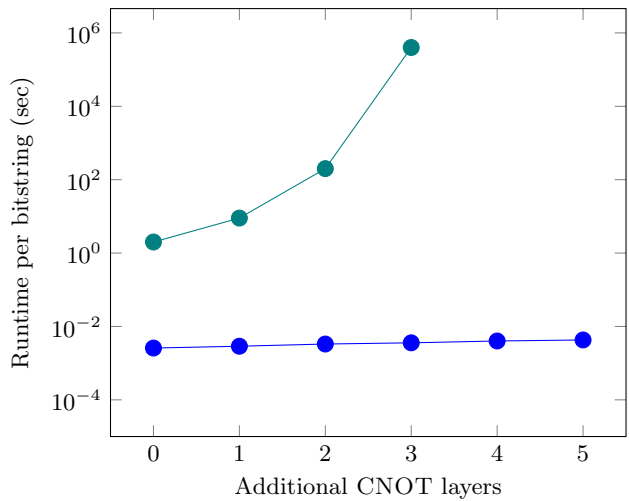

We ran 48- and 96-qubit simulations on a c2d-highcpu-112 instance \cite{c2d-highcpu-11} from the Google Cloud platform.  We recorded the simulation runtime for determining the amplitude of a random bitstring output $\ket{y}$ for a given number of trials.  The results are shown in \tab{performance-results-average}.  The reported times are average; we observed the worst times of $0.00451315$ and $8.16629$ seconds for the 48- and 96-qubit simulations, correspondingly.  The running time of $\fortyeightqubittime$ seconds during strong classical simulation beats all figures of merit that the physical quantum computer possesses, including the unitary execution time of ${\sim}0.01$ seconds, unitary execution with state preparation and measurement time of ${\sim}0.3$ seconds, and roughly 3 days including postselection and error correction \cite{bluvstein}.  Classical simulation cost considered on the per hour average price basis of $\$4.5767$ \cite{c2d-highcpu-11}, amounts to $\$0.00000327929$ (US Dollar) per amplitude in a 48-qubit computation and $\$0.00529662$ per amplitude in a 96-qubit computation, and may likely be difficult to meet by the quantum computer utilized \cite{bluvstein2023logical} due to power consumption alone.

We also ran 48- and 96-qubit simulations with up to 5 additional randomly oriented and asymmetric $\cnotgate$ and A/B diagonal block layers, and observed only a marginal difference in the simulation runtime.  For the 48-qubit case, we make an explicit comparison to \cite[Figure 5d bottom]{bluvstein2023logical}, see \fig{fig5dcomparison}.  Our simulation time does not grow in a meaningful way with the introduction of new circuit layers, suggesting that adding more layers to the $HQ$ circuit is unlikely to lead to a quantum advantage. 

\begin{table}[htbp]
    \centering
    \begin{tabular}{|c|c|c|}
        \hline
        \hspace{0.2cm} Number of qubits \hspace{0.2cm} & \hspace{0.2cm} Runtime in seconds \hspace{0.2cm} & \hspace{0.2cm} Number of trials \hspace{0.2cm} \\
        \hline
        48 & \fortyeightqubittime & 10000 \\
        96 & \ninetysixqubittime & 1000 \\
        \hspace{0.2cm} 96, with five extra $\cnotgate$ layers \hspace{0.2cm} & 4.18656 & 1000 \\
        \hline
    \end{tabular}
    \caption{Average-case runtime performance on a Google Cloud c2d-highcpu-112 instance for determining the amplitude of a random output string.}
    \label{tab:performance-results-average}
\end{table}

Simulating the 192-qubit case requires more processor cores than the 48- and 96-qubit cases, but is within striking distance of the amount of computational resources used for training deep neural networks. Preliminary analysis suggests that a \href{https://cloud.google.com/tpu/docs/system-architecture-tpu-vm}{TPU slice} from the Google Cloud platform could simulate a 192-qubit instance in a few days. The core workload is to iterate over all $x^R$ bit patterns and simulate the ${\sim}25\%$ Clifford circuits with nonzero amplitude contributions. This comes out to a total workload of 18,446,744,073,709,551,616 bitstrings to test, and 4,611,686,018,427,387,904 Clifford circuits to simulate.  The Clifford circuits are simulated with a linear system on a $128{\times}128$ dense binary matrix, the ideal workload for a TPU cluster.  We propose an implementation where CPUs test bitstrings for non-null amplitudes and upload a matrix of the non-null bitstrings for TPUs to simulate in parallel. 
Assuming a 6-instruction implementation of the bitstring test on a CPU, 75\% branch prediction success, a 10-instruction pipeline, all memory in cache or in registers (representing an amortized 5 clock cycles per instruction), and 2 fully occupied hyperthreads per core; a fleet of 1000 62-core Intel Xeon processors with 3GHz clocks can test 18,446,744,073,709,551,616 bitstrings in ${\sim}40$ hours. With access to 3,072 TPUv4 chips, one TPU slice can solve 12,288 Clifford circuits every ${\sim}1ns$ TPU clock cycle, as each TPUv4 chip has access to four $128{\times}128$ binary matrix logic cores \cite{jouppi2023tpu}.  This means our TPU slice can simulate all Clifford circuits in ${\sim}104$ hours.  Since the TPUs and CPUs operate in parallel, and the Clifford simulation test on TPUs is the bottleneck, the total time for our cluster of 1000 Google Cloud VMs and 4096 TPUs is on the order of magnitude of 104 hours, or approximately 4 days---within the range of expected cluster size, expense, and time expected of training a deep neural network.

We note that the above resource estimate for a $192$-qubit simulation reaches the number of qubits close to $208$, considered sufficient for quantum supremacy \cite{dalzell2020many} by IQP circuits.  This serves as evidence that the $HQ$ circuit is not as complex as an arbitrary/random IQP circuit.

\section{Discussion}

The $HQ_k$ circuit \cite{bluvstein2023logical} can be modified in various ways that may increase the difficulty of its classical simulation.  First, we discuss directions in which a further modification is unlikely to be fruitful as a means of increasing the complexity of classical simulation and is thus unlikely to lead to a demonstration of quantum advantage and next a direction that could be suitable. 
\begin{enumerate}
    \item Adding or removing $\zgate$ or $\czgate$ gates (\cite[Figure 5c]{bluvstein2023logical}) does not affect the simulation complexity in a meaningful way. This is because $\zgate$ gates can be accounted for with little complexity, and $\czgate$ gates are a part of the -H-CZ-H- Clifford circuit being simulated.
    \item Due to the use of the [[8,3,2]] error correcting code, each IQP circuit of the flavor considered in \cite{bluvstein2023logical} would have to apply $\czgate$ and $\cczgate$ gates to the triples of qubits in a qubit block, where no two blocks overlap.  This means that the covering set will contain $n/3$ qubits, thus allowing a simulation at the cost $O(\text{poly}(n)2^{n/3})$, offering an almost cubic improvement over straightforward simulations with complexity $O(\text{poly}(n)2^{n})$, such as the state vector simulation. 
    \item As demonstrated in \sec{runtime}, increasing the computation depth by adding layers of $\cnotgate$ and diagonal gates does not make the simulation meaningfully more complex, since the additional layers are subject to the conditions studied in \lem{1} and \lem{cliff}.  Thus, the complexity of the algorithm remains the same.
    \item In \sec{runtime}, we showed that doubling (to $96$) or quadrupling (to $192$) the number of qubits in the original Harvard/QuEra experiment does not disable the ability to simulate these kinds of computations using classical hardware.  This argument appears to first break when the number of logical qubits reaches $384$ or higher, which is far in excess of around $70$ to $105$ qubits marking a point where a quantum computation is known that may be intractable to reproduce by the classical means (\cite{nam2019low} Heisenberg Hamiltonian and \cite{babbush2018encoding} Hubbard model).  Given the lower bound on the number of qubits to imply classical intractability is $54$ \cite{pednault2019leveraging}, a more fair comparison can be between $16{=}70{-}54$ (or $50{=}105{-}54$) and $330=384{-}54$ additional qubits on top of absolute minimum, and such a difference is significant.  This comparison implies that a more fruitful direction to demonstrating quantum advantage can be the development of a complete logical library and full-fledged fault tolerance rather than increasing the number of qubits. 
\end{enumerate} 

We note the ability to implement a ``permutation $\cnotgate$'' \cite{bluvstein2023logical, wang2023fault} performing the transformation $\ket{a,b} \mapsto \ket{a,a \oplus b}$, equivalent to that offered by the $\cnotgate(a,b)$ gate, to a pair of qubits in a single logical block of the [[8,3,2]] code. We suspect that the experiment does not utilize such ``permutation $\cnotgate$s'' because this would require selective control of individual physical qubits rather than just selective control of logical blocks. Implementing a ``permutation $\cnotgate$'' fault-tolerantly appears to require commuting the corresponding physical $\swapgate$ gate towards the end of the circuit. This shuffles controls and targets in layers of transversal $\cnotgate$s and appears to require instructions outside the set that ``can be delivered in parallel with only a few control lines'' \cite{bluvstein2023logical}.  However, we believe this to be a potentially important gate to add to the instruction set, as, with sufficient effort, it allows to enlarge the minimal covering set considered in this work to contain all $n$ qubits, thus eliminating the cubic reduction in classical simulation time offered by our work.  Moreover, an arbitrary degree-3 phase polynomial can be implemented once such a ``permutation $\cnotgate$'' gate is added to the instruction set.  Specifically, suppose we want to implement $\cczgate(x,y,z)$, where all qubits belong to different blocks. WLOG, assume that the qubit $z\,{=}\,z^B$ belongs to the blue set in its block of three, $(z^R,z^G,z^B)$. Use ``permutation $\cnotgate$'' to move $x$ to red, and $y$ to green positions within their blocks, and implement four $Z$-angle products, $\cczgate(x\, {\oplus} \,z^R, y\, {\oplus}\, z^G, z^B)$, $\cczgate(x \,{\oplus} \,z^R, z^G, z^B)$, $\cczgate(z^R, y\, {\oplus}\, z^G, z^B)$, and $\cczgate(z^R, z^G, z^B)$, which, combined, compute the desired $\cczgate(x, y, z^B) = \cczgate(x, y, z)$. This is enabled by the cross-block addressable $\cnotgate$ gate constructed by combining transversal and ``permutation $\cnotgate$'' gates as follows: 
\begin{equation*}
	\Qcircuit @C=0.5em @R=0.5em @!{
		& \ctrl{4}  & \qw \\ 
		& \qw       & \qw \\ 
		& \qw       & \qw \\ \\
		& \targ     & \qw \\
		& \qw       & \qw \\ 
		& \qw       & \qw }
	\hspace{1em}\raisebox{-4em}{=}\hspace{1em}
	\Qcircuit @C=0.2em @R=0.5em @!{
		& \ctrl{4}  & \qw       & \qw  
		& \qw & \qw 
		& \ctrl{4}  & \qw       & \qw
		& \qw & \qw 
		& \ctrl{4}  & \qw       & \qw 
		& \qw & \qw & \qw \\ 
		& \qw       & \ctrl{4}  & \qw
		& \targ & \ctrl{1} 
		& \qw       & \ctrl{4}  & \qw   
		& \targ & \ctrl{1} 
		& \qw       & \ctrl{4}  & \qw      
		& \targ & \ctrl{1} & \qw \\ 
		& \qw       & \qw       & \ctrl{4}  
		& \ctrl{-1} & \targ 
		& \qw       & \qw       & \ctrl{4} 
		& \ctrl{-1} & \targ 
		& \qw       & \qw       & \ctrl{4} 
		& \ctrl{-1} & \targ & \qw \\ \\
		& \targ     & \qw       & \qw   
		& \qw & \qw 
		& \targ     & \qw       & \qw
		& \qw & \qw 
		& \targ     & \qw       & \qw 
		& \qw & \qw & \qw \\
		& \qw       & \targ     & \qw    
		& \qw & \qw 
		& \qw       & \targ     & \qw
		& \qw & \qw 
		& \qw       & \targ     & \qw 
		& \qw & \qw & \qw \\ 
		& \qw       & \qw       & \targ 
		& \qw & \qw 
		& \qw       & \qw       & \targ
		& \qw & \qw 
		& \qw       & \qw       & \targ 
		& \qw & \qw & \qw }
\end{equation*}
Note that our focus in the above discussion was on the demonstration of the ability to implement arbitrary degree-3 phase polynomials but not the efficiency of such an implementation.  We highlight that the complexity of the problem of simulating IQP circuits computable by arbitrary degree-3 polynomials is \#P-hard \cite[Theorem 1]{ehrenfeucht1990computational}, \cite[Theorem 6]{dalzell2020many}.

Our simulation is symbolic and thus does not suffer from the loss of precision in dealing with the floating point numbers.  We also found a certain swapping symmetry where pairs of blocks of qubits serving as controls/targets in the $HQ$ circuit can all be simultaneously interchanged allowing to extend the computed amplitude value for an outcome $\ket{y}$ to up to $n/6$ other outcomes $\ket{y^\prime}$ obtained from the $\ket{y}$ by permuting its bits.  We did not explore this symmetry further.


\end{document}